\documentclass[journal,onecolumn]{IEEEtran}

\usepackage{mathrsfs}
\usepackage{color}
\usepackage{enumerate}
\usepackage{mathrsfs}
\usepackage{amssymb}
\usepackage{amsfonts}
\usepackage{amsmath}
\usepackage{multirow}

\newtheorem{theorem}{Theorem}[section]

\newtheorem{lemma}[theorem]{Lemma}

\newtheorem{problem}{Problem}
\newtheorem{conjecture}[problem]{Conjecture}

\newtheorem{remark}{Remark}[section]
\newcommand{\F}{\mathbb{F}}

\newcommand{\GFn}{\text{GF}(p^n)}

\newcommand{\GFtlm}{\text{GF}(2^{2lm})}

\newcommand{\GFtm}{\text{GF}(2^{2m})}
\newcommand{\al}{\alpha}
\newcommand{\be}{\beta}

\newcommand{\om}{\omega}

\newcommand{\Sd}{\text{S}_d(z)}
\newcommand{\Cd}{\text{C}_d(z)}
\newcommand{\nin}{n_{\infty}(z)}
\newcommand{\nino}{n_{\infty}(z_1)}
\newcommand{\nint}{n_{\infty}(z_2)}
\newcommand{\nz}{n_{0}(z)}
\newcommand{\nzo}{n_{0}(z_1)}
\newcommand{\nzt}{n_{0}(z_2)}
\newcommand{\no}{n_{1}(z)}
\newcommand{\noo}{n_{1}(z_1)}
\newcommand{\nott}{n_{1}(z_2)}
\newcommand{\Tr}{\text{Tr}}

\begin{document}

\title{Some New Results on the Cross Correlation of $m$-sequences}

\author{Tao Zhang, Shuxing Li, Tao Feng and Gennian Ge
\thanks{The research of T. Feng was supported by Fundamental Research Fund for the Central Universities of China,
Zhejiang Provincial Natural Science Foundation under Grant LQ12A01019, the National Natural Science Foundation
of China under Grant 11201418, and the Research Fund for Doctoral Programs from the Ministry of Education of
China under Grant 20120101120089. The research of G. Ge was supported by the National Natural Science Foundation of China under Grant No.~61171198 and Zhejiang Provincial Natural Science Foundation of China under Grant No.~LZ13A010001.}
\thanks{T. Zhang is with the Department of Mathematics, Zhejiang University,
Hangzhou 310027,  China (e-mail: tzh@zju.edu.cn).}
\thanks{S. Li is with the Department of Mathematics, Zhejiang University,
Hangzhou 310027,  China (e-mail: sxli@zju.edu.cn).}
\thanks{T. Feng is with the Department of Mathematics, Zhejiang University,
Hangzhou 310027,  China (e-mail: tfeng@zju.edu.cn).}
\thanks{G. Ge is with  the School of Mathematical Sciences, Capital Normal University,
Beijing 100048, China (e-mail: gnge@zju.edu.cn).}
}

\maketitle

\begin{abstract}
The determination of the cross correlation between an $m$-sequence and its decimated sequence has been a long-standing research problem. Considering a ternary $m$-sequence of period $3^{3r}-1$, we determine the cross correlation distribution for decimations $d=3^{r}+2$ and $d=3^{2r}+2$, where $\gcd(r,3)=1$. Meanwhile, for a binary $m$-sequence of period $2^{2lm}-1$, we make an initial investigation for the decimation  $d=\frac{2^{2lm}-1}{2^{m}+1}+2^{s}$, where $l \ge 2$ is even and $0 \le s \le 2m-1$. It is shown that the cross correlation takes at least four values. Furthermore, we confirm the validity of two famous conjectures due to Sarwate et al. and Helleseth in this case.
\end{abstract}

\begin{keywords}
cross correlation, cross correlation distribution, decimation, $m$-sequences, Weil sums
\end{keywords}

\section{Introduction}

During the last decades, many applications of sequences with low correlation have been found in cryptography, radar and wireless communication systems \cite{GG05}. In the CDMA system, a popular method to spread the spectrum is the use of sequences. With the low correlation property, the interference of different users during the transmission can be minimized. Therefore, sequences with good correlation properties have been an important research problem enjoying considerable interests \cite{HK98}.

Let $p$ be a prime. Let $\{a_t\}$ and $\{b_t\}$ be two sequences of period $N$ with elements from a finite field $\textup{GF}(p)$. The cross correlation between $\{a_t\}$ and $\{b_t\}$ at shift $\tau$ is defined by
$$C_{a,b}(\tau)=\sum_{t=0}^{N-1}\omega^{a_{t+\tau}-b_{t}},$$
where $0 \le \tau < N$ and $\omega$ is a complex $p$-th root of unity.

As a pseudorandom sequence with ideal two-level auto correlation, the maximal linear sequence ($m$-sequence) has attracted much attention. Many researches focus on the cross correlation properties, since low cross correlation guarantees good correlation properties for $m$-sequences (see \cite{TH76,CD96,CCD00,DHKM01,HX01,C04,DFHR06,K12} and the references therein).

Recall that the trace function from finite field $E=\textup{GF}(p^{n})$ onto its subfield $F=\textup{GF}(p^{r})$ is defined by
$$Tr_{r}^{n}(x)=x+x^{p^{r}}+x^{p^{2r}}+...+x^{p^{n-r}}.$$
For $r=1$, we get the absolute trace function mapping onto the prime field $GF(p)$, which is denoted by $Tr_{n}$ or $Tr$. A $p$-ary $m$-sequence $\{a_{t}\}$ of period $p^n-1$ can be represented by
$$
a_t=\Tr(\be\al^t), \quad 0 \le t \le p^n-2,
$$
where $\be \in \GFn^*$ and $\al$ is a primitive element of $\GFn$.

Suppose $(d,p^n-1)=1$. The $d$-decimation of $\{a_{t}\}$, which is denoted by $\{a_{dt}\}$, is also an $m$-sequence with the same period. Note that if $d \in \{1,p,\ldots,p^{n-1}\}$, $\{a_{dt}\}$ is simply a cyclic shift of $\{a_{t}\}$. The cross correlation between $\{a_{t}\}$ and $\{a_{dt}\}$ takes two values and is easy to compute \cite[Theorem 3.1]{TH76}. Below, we always consider the nondegenerate decimation $d$ where $d \not\in \{1,p,\ldots,p^{n-1}\}$. The cross correlation between an $m$-sequence of period $p^{n}-1$ and its $d$-decimation can be described by
\begin{align*}
C_d(\tau) &= \sum_{t=0}^{p^n-2} \om^{a_{t+\tau}-a_{dt}} \\
          &= -1+\sum_{x \in \GFn} \chi(\al^{\tau}x-x^d), \\
\end{align*}
where $\chi(x)=w^{\Tr(x)}$ for any $x \in \GFn$ and $0 \le \tau \le p^n-2$. Clearly, calculating the cross correlation value is to compute the Weil sum
$$
\Cd=\sum_{x \in \GFn^*} \chi(zx-x^d),
$$
where $z \in \GFn^*$. Hence, computing the cross correlation distribution is to determine the multiset
$$
\{\Cd \, \mid \, z \in \GFn\}.
$$
Noting that the cross correlation distribution essentially arises in many other contexts with various names, please refer to the appendix of \cite{K12} for more details.

For the cross-correlation function between an $m$-sequence and its $d$-decimation, an overview of known results can be found in \cite{TH76,DFHR06,CN12}. Besides, further generalizations have been made to study the cross correlation of an $m$-sequence and its $d$-decimated sequence with $\gcd(d,p^{n}-1)>1$ (see \cite{NH06,NH062,NH07,SKNS08,HHKZLJ09,XZH10}). When $p=2$ and $(d,2^{n}-1)=1$, the known results on the cross correlation distribution of an $m$-sequence and its decimation are listed in Table~\ref{table1}, where $v_2(k)$ is the largest power of $2$ dividing $k$. Meanwhile, Table~\ref{table2} summaries the known results on the cross correlation distribution when $p$ is an odd prime.

There are many methods which have been proposed to determine the cross correlation distribution. With the help of some known exponential sums, Helleseth \cite{TH76,TH03} computed the cross correlation distribution for several decimations. Luo and Feng \cite{LF08} used the technique of quadratic forms to attack this problem. In \cite{DFHR06}, Dobbertin et al. developed a delicate method involving the use of Dickson polynomials.

In this paper, we consider the cross correlation between a ternary $m$-sequence of period $3^{3r}-1$ and its $d$-decimation with $d=3^r+2$ or $d=3^{2r}+2$, where $(r,3)=1$. Following the idea of Dobbertin \cite{D98} and Feng et al. \cite{FLX13}, we completely determine the cross correlation distribution. Besides, for the binary $m$-sequence of period $2^{2lm}-1$ and decimation $d=\frac{2^{2lm}-1}{2^{m}+1}+2^{s}$, with $l \ge 2$ being even and $0 \le s \le 2m-1$, we obtain some results on the cross correlation values. When $l$ is odd, the decimation $d$ is of Niho type, which has been extensively studied \cite{YN70,DHKM01,HX01,C04,HR05,DFHR06}. Recall that any nondegenerate decimation leads to at least three cross correlation values \cite[Theorem 4.1]{TH76}. We further prove that the cross correlation takes at least four values for this decimation. While it seems pretty hard to determine the cross correlation distribution, we confirm the validity of the following two famous conjectures due to Sarwate et al. \cite{SP80} and Helleseth \cite{TH76} respectively. Below, we define $\Sd=\Cd+1$.

\begin{conjecture}{\cite{SP80}}\label{7}
Let $n=2t$ and $p=2$, then $\Sd \geq 2^{t+1}.$
\end{conjecture}

\begin{conjecture}{\cite[Conjecture 5.1]{TH76}}\label{8}
If $p^{n}>2$ and $d\equiv1\pmod{p-1},$ then $\Sd=0$ for some $z\in\GFn^{*}.$
\end{conjecture}

This paper is organized as follows. In Section $2$, we determine the cross correlation distribution for a ternary $m$-sequence and its decimated sequence mentioned above. In Section $3$, we present some results on the cross correlation between a binary $m$-sequence and its $d$-decimation, where the above two conjectures are verified for $d=\frac{2^{2lm}-1}{2^{m}+1}+2^{s}$, with $l \ge 2$ being even and $0 \le s \le 2m-1$. Section $4$ concludes the paper.

\begin{table}[h]\renewcommand{\arraystretch}{1.7}

\begin{center}
\caption{Cross correlation distribution between a binary $m$-sequence of period $2^n-1$ and its $d$-decimation with $(d,2^{n}-1)=1$}
\begin{tabular}{|c|c|c|}
\hline
Sources  &  $d$  &  $\sharp$ cro-co \\ \hline
Gold \cite{G68} &  $d=2^{k}+1$,$\frac{n}{(n,k)}$ is odd  & 3 \\ \hline
Kasami \cite{K71} & $d=2^{2k}-2^{k}+1$, $\frac{n}{(n,k)}$ is odd  & 3 \\ \hline
Welch \cite{CCD00} & $d=2^{k}+3$, $n=2k+1$  & 3\\ \hline
\multirow{2}{*}{Hollmann and Xiang \cite{HX01}} & $d=2^{2k}+2^{k}-1$, $k=\frac{n-1}{4}$ if $n\equiv1\pmod{4}$
&\multirow{2}{*}{3}\\& and $k=\frac{3n-1}{4}$ if $n\equiv3\pmod{4}$ & \\ \hline
Cusick and Dobbertin \cite{CD96} & $d=2^{k}+2^{\frac{k+1}{2}}+1$, $n=2k$, $k$ is odd & 3 \\ \hline
Cusick and Dobbertin \cite{CD96} & $d=2^{k+1}+3$, $n=2k$, $k$ is odd & 3 \\ \hline
Niho \cite{YN70} & $d=2^{2k+1}-1$, $n=4k$ & 4 \\ \hline
Niho \cite{YN70} & $d=(2^{2k}+1)(2^{k}-1)+2$, $n=4k$ & 4 \\ \hline
Dobbertin \cite{D98} & $\sum_{i=0}^{2k}2^{im}$, $n=4k$, $0<m<n$, $\gcd(m,n)=1$ & 4 \\ \hline
Helleseth and Rosendahl \cite{HR05} & $d=(2^{2k}+2^{s+1}-2^{k+1}-1)/(2^{s}-1)$, $n=2k$, $2s|k$ & 4 \\ \hline
\multirow{2}{*}{Dobbertin et al. \cite{DFHR06}} & $d=(2^{k}-1)s+1$, $s\equiv2^{r}(2^{r}\pm1)^{-1}\pmod{2^{k}+1}$,
&\multirow{2}{*}{4}\\& $v_{2}(r)<v_{2}(k)$ & \\ \hline
Helleseth \cite{TH76} & $d=2^{k}+3$, $n=2k$ & 5 \\ \hline
Dobbertin \cite{D98} & $d=2^{2k}+2^{k}+1$, $n=4k$, $k$ is odd & 5 \\ \hline
Johansen and Helleseth \cite{JH09} & $d=\frac{5}{3}$, $n$ is odd & 5 \\ \hline
Johansen et al. \cite{JHK09} & $d=\frac{17}{5}$, $n$ is odd & 5 \\ \hline
Boston and McGuire \cite{BM10} & $d=11$, $n$ is odd & 5 \\ \hline
Helleseth \cite{TH1978} & $d=2^{2k}-2^{k}+1$, $n=4k$, $k$ is even & 6 \\ \hline
\multirow{2}{*}{Helleseth \cite{TH76}} & $d=\frac{1}{3}(2^{n}-1)+2^{s}$, $n$ is even,
&\multirow{2}{*}{6}\\& $s<n$ and $\frac{1}{3}2^{-s}(2^{n}-1)\neq2\pmod{3}$ & \\ \hline
Dobbertin et al. \cite{DFHR06} & $d=3\cdot2^{k-1}-1$, $n=2k$, $k$ is odd & 6 \\ \hline
Feng et al. \cite{FLX13} & $d=2^{t+1}+3$, $n=2t$, $t\equiv2\pmod{4}$ and $t\geq6$ & 7 \\ \hline
\end{tabular}
\label{table1}
\end{center}
\end{table}

\begin{table}[h]\renewcommand{\arraystretch}{2.2}

%\small
\begin{center}
\caption{Cross correlation distribution between a non-binary $m$-sequence of period $p^{n}-1$ and its $d$-decimation}
\begin{tabular}{|c|c|c|c|c|c|}
\hline
Sources  & $p$  & $n$    &  $d$ & $\gcd(d,p^{n}-1)$ &  $\sharp$ cro-co \\ \hline
Helleseth \cite{TH76}&odd prime &any&  $d=p^{2k}-p^{k}+1$,$\frac{n}{(n,k)}$ is odd &1 & 3 \\ \hline
Helleseth \cite{TH76}   &odd prime &any&  $d=\frac{1}{2}(p^{2k}+1)$, $\frac{n}{(n,k)}$ is odd &1 & 3 \\ \hline
Dobbertin et al. \cite{DHKM01} &$3$ &odd&  $d=2\cdot3^{\frac{n-1}{2}}+1$ &1 & 3 \\ \hline
Helleseth \cite{TH76}   &$p^{\frac{n}{2}}\not\equiv 2\pmod{3}$ &even&  $d=2p^{\frac{n}{2}}-1$ &1& 4 \\ \hline
Helleseth \cite{TH76}   &$p^{n}\equiv 1\pmod{4}$  &any&  $d=\frac{1}{2}(p^{n}-1)+p^{i}$, $0\leq i<n$   &1& 5 \\ \hline
\multirow{2}{*}{Helleseth \cite{TH76}} & $p\equiv 2\pmod{3}$&even&  $d=\frac{1}{3}(p^{n}-1)+p^{i}$, $0\leq i<n$, &1& \multirow{2}{*}{}{6}\\&&&$\frac{1}{2}p^{-i}(p^{n}-1)\not\equiv 2\pmod{3}$ && \\ \hline
Helleseth \cite{TH03} & $p^{m}\neq2\pmod{3}$ & $0\pmod{4}$ & $d=p^{2m}-p^{m}+1$, $n=4m$ & 1 & 6 \\ \hline
Luo and Feng \cite{LF08}   &odd prime &any& $d=\frac{p^{k}+1}{2}$, odd $k/e$, $e=\gcd(n,k)$ &variable& variable \\ \hline
Seo et al. \cite{SKNS08}   &odd prime &$0\pmod{4}$& $d=(\frac{p^{m}+1}{2})^{2}$, $n=2m$ &$\frac{p^{m}+1}{2}$& 4 \\ \hline
Choi et al. \cite{CKN12} & $3\pmod{4}$ & odd & $d=\frac{p^{n}+1}{p^{k}+1}+\frac{p^{n}-1}{2}$, $k|n$ & 2 & 9 \\ \hline
\end{tabular}
\label{table2}
\end{center}
\end{table}

\section{The cross correlation distribution for the ternary $m$-sequence}

In this section, for the ternary $m$-sequence of period $3^{3r}-1$, we determine the cross correlation distribution with decimation $d=3^{r}+2$ or $d=3^{2r}+2$, where $(r,3)=1$.

At first, we recall two results which can be found in \cite{LN83}.

\begin{lemma}{\cite[Theorem 5.15]{LN83}}\label{1}
Let $\textup{GF}(q)$ be a finite field with $q=p^{s}$, where $p$ is an odd prime and $s$ is a positive integer. Let $\eta$ be the quadratic character of $\textup{GF}(q)$ and let $\chi$ be the canonical additive character of $\textup{GF}(q)$. Then
 \[G(\eta,\chi)=\begin{cases}(-1)^{s-1}q^{\frac{1}{2}};\textup{ if } p\equiv1\pmod{4},\\
(-1)^{s-1}i^{s}q^{\frac{1}{2}};\textup{ if }p\equiv3\pmod{4},\end{cases}\]
where $G(\eta,\chi)$ is the Gauss sum related to $\eta$ and $\chi$.
\end{lemma}

\begin{lemma}\cite[Theorem 5.33]{LN83}\label{2}
Let $\chi$ be a nontrivial additive character of $GF(q)$ with $q$ odd, and let $f(x)=a_{2}x^{2}+a_{1}x+a_{0}\in GF(q)[x]$ with $a_{2}\neq0$. Then
$$\sum_{c\in GF(q)}\chi(f(c))=\chi(a_{0}-a_{1}^{2}(4a_{2})^{-1})\eta(a_{2})G(\eta,\chi),$$
where $\eta$ is the quadratic character of $GF(q)$.
\end{lemma}

The following well known identities can be found in \cite{TH76}.
\begin{lemma}\label{3}
We have
$$\sum_{z\in GF(p^{n})}\Sd =p^{n},$$
$$\sum_{z\in GF(p^{n})}\Sd^{2}=p^{2n},$$
$$\sum_{z\in GF(p^{n})}\Sd^{3}=p^{2n}b_{3},$$
where $b_{3}$ is the number of common solutions of
$$x+y+1=0,$$
$$x^{d}+y^{d}+1=0,$$
such that $x,y\in GF(p^{n}).$
\end{lemma}

As a preparation, we have the following lemma.

\begin{lemma}\label{6}
Given an integer $r$ with $\gcd(r,3)=1$. Suppose $n=3r$, $d=3^{r}+2$ or $d=3^{2r}+2$.
Then for $x,y\in GF(3^{n})$, the number of common solutions of
$$x+y+1=0,$$
and
$$x^{d}+y^{d}+1=0,$$
is $3^r$.
\end{lemma}
\begin{proof}
We deal with the case where $d=3^r+2$. When $d=3^{2r}+2$, the proof is similar.
Note the above equations are equivalent to
$$(x+1)^{d}-x^{d}-1=0.$$
Consequently,
$$(x+1)^{3^{r}}(x+1)^{2}-x^{3^{r}+2}-1=0,$$
which leads to
$$(x-1)(x^{3^{r}}-x)=0.$$
Hence, we deduce that $x\in GF(3^{r})$, which means there are $3^r$ common solutions.
\end{proof}

Now we state our main result.

\begin{theorem}\label{4}
Given an integer $r \ge 2$ with $\gcd(r,3)=1$. Set $n=3r$, $d=3^{r}+2$ or $d=3^{2r}+2$. For the ternary $m$-sequence of period $3^n-1$, the cross correlation with its $d$-decimation is listed as follows. When $r$ is even, the cross correlation distribution is
\begin{center}
\begin{tabular}{c c c c}
$-1$  &  occurs  &  $\frac{3^{3r}+3^{2r}}{2}-3^{r}$   &  times\\
$3^{2r}-1$  &  occurs  &  $3^{r}$   &  times\\
$3^{\frac{3r}{2}}-1$  &  occurs  &  $\frac{3^{3r-1}-3^{2r-1}}{2}$   &  times\\
$-3^{\frac{3r}{2}}-1$  &  occurs  &  $\frac{3^{3r-1}-3^{2r-1}}{2}$   &  times\\
$2\cdot3^{\frac{3r}{2}}-1$  &  occurs  &  $\frac{3^{3r-1}-3^{2r-1}}{4}$   &  times\\
$-2\cdot3^{\frac{3r}{2}}-1$  &  occurs  &  $\frac{3^{3r-1}-3^{2r-1}}{4}$   &  times\\
\end{tabular}
\end{center}
When $r$ is odd, the cross correlation distribution is
\begin{center}
\begin{tabular}{c c c c}
$-1$  &  occurs  &  $2\cdot3^{3r-1}+3^{2r-1}-3^{r}$   &  times\\
$3^{2r}-1$  &  occurs  &  $3^{r}$   &  times\\
$3^{\frac{3r+1}{2}}-1$  &  occurs  &  $\frac{3^{3r-1}-3^{2r-1}}{2}$   &  times\\
$-3^{\frac{3r+1}{2}}-1$  &  occurs  &  $\frac{3^{3r-1}-3^{2r-1}}{2}$   &  times\\
\end{tabular}
\end{center}
\end{theorem}
\begin{proof}
In the following, we only prove the case $d=3^{r}+2$. The case $d=3^{2r}+2$ can be handled similarly. We fix $d=3^{r}+2$, $E=GF(3^{n})$, $F=GF(3^{r})$ and $n=3r$. It is routine to verify $\gcd(d,3^{n}-1)=1$.

Let $a$ be a primitive element of $GF(27)$ with
$$a^{3}+2a+1=0.$$
Since $\gcd(r,3)=1$, we have $E=F(a)$. For any $x \in E$, it can be expressed as
$$x=x_{0}+x_{1}a+x_{2}a^{2},$$
where $x_0,x_1,x_2 \in F$.

Since $\gcd(r,3)=1$, we consider the case $r\equiv2\pmod{3}$ at first, in which $a^{3^{r}}=a^{9}$. The first step is to compute a direct representation of $Tr_{n}(x^{d})$ as a function of $x_0,x_1$ and $x_2$. Note that
$Tr_{r}^{n}(1)=Tr_{r}^{n}(a)=0$ and $Tr_{r}^{n}(a^{2})=2$. A lengthy routine computation shows
$$Tr_{n}(x^{d})=Tr_{r}(x_{1}x_{2}^{2}+x_{0}x_{2}^{2}+2x_{1}^{2}x_{2}+2x_{1}).$$
Next, we compute $\Cd$ for some fixed $z\in E$. Putting
$$z=z_{0}+z_{1}a+z_{2}a^{2}$$
with $z_0,z_1,z_2 \in F$, we find
$$Tr_{n}(xz)=Tr_{r}(2x_{2}z_{2}+2x_{0}z_{2}+2x_{1}z_{1}+2x_{2}z_{0}).$$
Consequently,
\begin{eqnarray*}
\Sd&=&\sum_{x_{0},x_{1},x_{2}\in F}\chi_{F}(x_{1}x_{2}^{2}+x_{0}x_{2}^{2}+2x_{1}^{2}x_{2}+2x_{1}+2x_{2}z_{2}+2x_{0}z_{2}+2x_{1}z_{1}+2x_{2}z_{0})\\
                    &=&\sum_{x_{0},x_{1},x_{2}\in F}\chi_{F}(x_{0}(x_{2}^{2}+2z_{2})+x_{1}x_{2}^{2}+2x_{1}^{2}x_{2}+2x_{1}+2x_{2}z_{2}+2x_{1}z_{1}+2x_{2}z_{0})\\
                    &=&3^{r}\sum_{x_{1}\in F, x_{2}\in M}\chi_{F}(x_{1}x_{2}^{2}+2x_{1}^{2}x_{2}+2x_{1}+2x_{2}z_{2}+2x_{1}z_{1}+2x_{2}z_{0}),\\
\end{eqnarray*}
where $M=\{x_{2}\in F| x_{2}^{2}=z_{2}\}.$

If $z_{2}=0$, then $M=\{0\}$. We have
\begin{align*}
\Sd &=3^{r}\sum_{x_{1}\in F}\chi_{F}(2x_{1}(1+z_{1}))\\
    &=\begin{cases}0 & \textup{ if } z_{1}\neq2,\\
3^{2r} & \textup{ if }z_{1}=2.\end{cases}
\end{align*}

If $z_{2}$ is a nonsquare in $F$, then $M=\emptyset$ and $\Sd=0.$

If $z_{2}$ is a nonzero square in $F$, let $z_{2}=b^{2},$ then $M=\{\pm b\}$. Hence,
\begin{align*}
\Sd &= 3^{r}\sum_{x_{1}\in F}\chi_{F}(2bx_{1}^{2}+(b^{2}+2z_{1}+2)x_{1}+2b^{3}+2bz_{0})\\
    & +3^{r}\sum_{x_{1}\in F}\chi_{F}(bx_{1}^{2}+(b^{2}+2z_{1}+2)x_{1}+b^{3}+bz_{0}).
\end{align*}
With the help of Lemma~\ref{1} and Lemma~\ref{2}, we deduce
$$\Sd=(-1)^{r-1}\cdot i^{r}\cdot3^{\frac{3r}{2}}(\eta(2b)\chi_F(c)+\eta(b)\chi_F(-c)),$$
where
$$c=2b^{3}+2bz_{0}-(b^{2}+2z_{1}+2)^{2}(2b)^{-1}.$$
Note that
 \[\eta(2)=\begin{cases}1 & \textup{ if } r \textup{ is even},\\
-1 & \textup{ if } r \textup{ is odd}.\end{cases}\]
Suppose $A=\eta(2b)\chi_F(c)+\eta(b)\chi_F(-c)$. Since $\chi_F(c)=\overline{\chi_F(-c)}$, we have
 \[A=\begin{cases}\pm1,\pm2;\textup{ if } r \textup{ is even},\\
0,\pm\sqrt{-3};\textup{ if } r \textup{ is odd}.\end{cases}\]

When $r$ is even, $\Sd$ takes six values $0$, $3^{2r}$, $3^{\frac{3r}{2}}$, $-3^{\frac{3r}{2}}$, $2\cdot3^{\frac{3r}{2}}$ and $-2\cdot3^{\frac{3r}{2}}$. For $1 \le i \le 6$, use $N_{i}$ to denote the number of occurrences in the corresponding order above. With Lemma~\ref{3}, Lemma~\ref{6} and the above discussion, we get
$$N_{1}=\frac{3^{3r}}{2}+\frac{3^{2r}}{2}-3^{r},$$
$$N_{2}=3^{r},$$
$$N_{1}+N_{2}+N_{3}+N_{4}+N_{5}+N_{6}=3^{3r},$$
$$3^{2r}N_{2}+3^{\frac{3r}{2}}(N_{3}-N_{4})+2\cdot3^{\frac{3r}{2}}(N_{5}-N_{6})=3^{3r},$$
$$3^{4r}N_{2}+3^{3r}(N_{3}+N_{4})+4\cdot3^{3r}(N_{5}+N_{6})=3^{6r},$$
$$3^{6r}N_{2}+3^{\frac{9r}{2}}(N_{3}-N_{4})+8\cdot3^{\frac{9r}{2}}(N_{5}-N_{6})=3^{7r}.$$
Thus, we complete the proof for the case $r\equiv2\pmod{6}$.

When $r$ is odd, $\Sd$ takes four values $0$, $3^{2r}$, $3^{\frac{3r+1}{2}}$ and $-3^{\frac{3r+1}{2}}$. For $1 \le i \le 4$, use $N_{i}$ to denote the number of occurrences in the corresponding order above. With Lemma~\ref{3}, Lemma~\ref{6} and the above discussion, we get
$$N_{2}=3^{r},$$
$$N_{1}+N_{2}+N_{3}+N_{4}=3^{3r},$$
$$3^{2r}N_{2}+3^{\frac{3r+1}{2}}(N_{3}-N_{4})=3^{3r},$$
$$3^{4r}N_{2}+3^{3r+1}(N_{3}+N_{4})=3^{6r}.$$
Thus, we complete the proof for the case $r\equiv5\pmod{6}$.

For the remaining case $r\equiv1\pmod{3}$, a similar discussion leads to
$$Tr_{n}(x^{d})=Tr_{r}(2x_{2}+x_{0}x_{2}^{2}+2x_{1}^{2}x_{2}+2x_{1}x_{2}^{2}+x_{1}).$$
The cross correlation distribution can be obtained in a similar way.
\end{proof}

\begin{remark}
When $r=3$, a numerical experiment shows that the cross correlation distribution for the ternary $m$-sequence of period $3^9-1$ with decimation $d=3^3+2=29$ or $d=3^6+2=731$ is
\begin{center}
\begin{tabular}{c c c c}
$-1$  &  occurs  &  $13338$   &  times\\
$728$  &  occurs  &  $27$   &  times\\
$242$  &  occurs  &  $3159$   &  times\\
$-244$  &  occurs  &  $3159$   &  times\\
\end{tabular}
\end{center}
This result is consistent with the distribution presented in Theorem~\ref{4}. Hence, in the case where $(r,3)=3$, we conjecture that the correlation distribution is the same as that in Theorem~\ref{4}.
\end{remark}

\section{Some results on the cross correlation of binary $m$-sequences}

In this section, we focus on the cross correlation between a binary $m$-sequence of period $2^{2lm}-1$ and its $d$-decimated sequence with $d=\frac{2^{2lm}-1}{2^{m}+1}+2^{s}$, where $0 \le s \le 2m-1$ and $(2^{s-1}-l,2^m+1)=1$. Note that $(2^{s-1}-l,2^m+1)=1$ is equivalent to $(d, 2^{2lm}-1)=1$. Some special cases of this form have been studied before. For example, when $m=1$, the decimation $d=\frac{2^{2l}-1}{3}+2^s$ has been studied in \cite{TH76} where the cross correlation distribution was obtained. If $l=2$ and $s=0$, the decimation $d=\frac{2^{4m}-1}{2^m+1}+1$ has also been investigated in \cite{TH1978}. In fact, when $l$ is odd, it is straightforward to verify that $d$ is of Niho type. As shown in \cite{C04}, for the decimation of Niho type, the cross correlation takes at least four values. Consequently, it is natural to ask if the same thing happens when $l$ is even. In this case, it is clear that $d$ may not be of Niho type. We will show that $\Cd$ also takes at least four values. In addition, we confirm that Conjecture~\ref{7} and Conjecture~\ref{8} are true for this type of decimation $d$.

Throughout the rest of this section, we always assume that $d=\frac{2^{2lm}-1}{2^{m}+1}+2^{s}$, where $0 \le s \le 2m-1$, $(2^{s-1}-l,2^m+1)=1$ and $l$ is even. Let $\al$ be a primitive element of $\GFtlm$. We define
\begin{align*}
C_{\infty} &= \{0\}, \\
C_{0}      &= \{\al^{j(2^m+1)} \, | \, 0 \le j \le \frac{2^{2lm}-1}{2^m+1}-1\}, \\
C_{1}      &= \GFtlm \setminus (C_{0} \cup C_{\infty}).
\end{align*}
The following lemma is a special case of \cite[Lemma 3.5]{TH76}.
\begin{lemma}
\begin{equation*}
\sum_{x \in \GFtlm} \chi(ax^{2^m+1})=
                   \begin{cases}
                       2^{2lm} & \text{if  $a\in C_{\infty}$,}\\
                       -2^{(l+1)m} & \text{if  $a\in C_{0}$,}\\
                       2^{lm} & \text{if $a\in C_{1}$.}
                   \end{cases}
\end{equation*}
\end{lemma}

Now, we are ready to prove our result.
\begin{theorem}\label{9}
Suppose $d=\frac{2^{2lm}-1}{2^{m}+1}+2^{s}$, where $0 \le s \le 2m-1$, $(2^{s-1}-l,2^m+1)=1$ and $l$ is even.  Then
\begin{enumerate}
\item[(i)] $\Sd=0$ for some $z \in \GFtlm^*$;
\item[(ii)]  $\Cd$ takes at least four values;
\item[(iii)] There exists a $z \in \GFtlm$ such that $\Sd \ge 2^{lm+1}$.
\end{enumerate}
\end{theorem}

\begin{proof}
By \cite[Theorem 3.8]{TH76}, we have
\begin{equation*}
\Cd=-1+\frac{1}{2^m+1}\sum_{j=0}^{2^m}\sum_{x\in\GFtlm}\chi(x^{2^m+1}(z\alpha^{j}+\alpha^{dj2^{-s}})),
\end{equation*}
where $2^{-s}$ is the inverse of $2^s$ modulo $2^{2lm}-1$.

For any $z\in \GFtlm$, define
$$n_{i}(z)=|\{j \mid 0\le j \le 2^m, z\alpha^{j}+\alpha^{dj2^{-s}}\in C_{i}\}|,$$
where $i=0,1,\infty.$ Thus,
\begin{equation}\label{ccbinary}
\Cd=-1+\frac{1}{2^{m}+1}(2^{2lm}\nin-2^{(l+1)m}\nz+2^{lm}\no).
\end{equation}
Equivalently, $\Sd=\frac{2^{lm}}{2^{m}+1}(2^{lm}\nin-2^{m}\nz+\no)$ and $2^{lm} \mid \Sd$. Then a direct application of \cite[Lemma 3]{C04} completes the proof of (i).

Set $A=\{\al^j \mid 0 \le j \le 2^m\}$. It follows from the definition of $n_i(z)$ that
\begin{align*}\label{equiv}
\nin &= |\{x \in A \mid zx+x^{d2^{-s}}=0\}|,\\
\nz  &= |\{x \in A \mid (zx+x^{d2^{-s}})^{\frac{2^{2lm}-1}{2^m+1}}=1\}|,\\
\nin+\nz+\no &= 2^m+1.
\end{align*}
Moreover, since $(d2^{-s}-1,2^{2lm}-1)=\frac{2^{2lm}-1}{2^m+1}$, there are exactly $2^m+1$ choices of $z \in \GFtlm$ such that $\nin=1$.

In the following, we will turn to the proof of (ii) and (iii). Since $\Sd$ attains $0$ and at least one negative value \cite[Lemma 1]{C04}, it suffices to show that $\Sd$ can take two distinct positive values. Below, we split our discussion into two cases with $l>2$ and $l=2$.

\textbf{Case 1: $ l>2 $}

Note that $\nin+\nz+\no=2^m+1$. When $\nin=1$, by (\ref{ccbinary}), we simply have $\Sd \ge \frac{1}{2^m+1} 2^{2lm}-2^{(l+2)m}=\frac{2^{(l+2)m}(2^{(l-2)m}-1)}{2^m+1}>2^{lm+1}$.

Next, we show that $\Cd$ takes at least four values. Otherwise, assume that $\Sd$ takes three values $\{u,v,0\}$, where $u>2^{lm+1}$ and $v<0$. If $\nin=0$, by (\ref{ccbinary}), we have $\Sd \le 2^{lm}$. Thus, $\Sd$ attains distinct values when $\nin=0$ and $\nin=1$. Therefore, given $z \in \GFtlm$, $\Sd=u$ if and only if $\nin=1$. We define
\begin{align*}
N_u &= |\{z \in \GFtlm \mid \Sd=u\}|,\\
N_v &= |\{z \in \GFtlm \mid \Sd=v\}|,\\
N_0 &= |\{z \in \GFtlm \mid \Sd=0\}|.
\end{align*}
Note that there are exactly $2^m+1$ choices of $z$ such that $\nin=1$. We get $N_u=2^m+1$. On the other hand, by the first two equations of Lemma~\ref{3}, we have
\begin{align*}
uN_u+vN_v &= 2^{2lm},\\
u^2N_u+v^2N_v &= 2^{4lm}.
\end{align*}
A direct computation shows that $N_u=\frac{2^{2lm}(v-2^{2lm})}{uv-v^2}$. Use $v_2(k)$ to denote the largest power of $2$ dividing $k$. Since $v_2(v)<2lm$, we simply have $v_2(2^{2lm}(v-2^{2lm}))=v_2(v)+2lm$ and $v_2(uv-v^2)=v_2(v)+v_2(u-v)$.

Below, we will show that $v_2(u-v) < 2lm$. Suppose $z_1 \in N_u$. Then $\nino=1, \nzo+\noo=2^m$ and
\begin{align*}
u=S_d(z_1) &= \frac{2^{lm}}{2^m+1}(2^{lm}-2^m\nzo+\noo)\\
           &= \frac{2^{lm}}{2^m+1}(2^{lm}+2^m(1-\nzo)-\nzo).
\end{align*}
Suppose $z_2 \in N_v$. Then $\nint=0, \nzt+\nott=2^m+1$ and
\begin{align*}
v=S_d(z_2) &= \frac{2^{lm}}{2^m+1}(-2^m\nzt+\nott)\\
           &= \frac{2^{lm}}{2^m+1}(2^m(1-\nzt)-\nzt+1).
\end{align*}
Hence,
$$
u-v=\frac{2^{lm}}{2^m+1}(2^{lm}+2^m(\nzt-\nzo)+\nzt-\nzo-1).
$$
If $\nzt-\nzo-1=0$, then $u-v=\frac{2^{lm}}{2^m+1}(2^{lm}+2^m)$ and $v_2(u-v)=(l+1)m < 2lm$. If $\nzt-\nzo-1 \ne 0$, since $0 \le \nzo \le 2^m$ and $0 \le \nzt \le 2^m+1$, we have $v_2(\nzt-\nzo) \le m$ and $v_2(\nzt-\nzo-1) \le m$. It is straightforward to verify that $2^m(\nzt-\nzo)+\nzt-\nzo-1 \ne 0$ and $v_2(2^m(\nzt-\nzo)+\nzt-\nzo-1) \le 2m$. Thus, $v_2(u-v) \le (l+2)m < 2lm$.

Consequently, we have $v_2(2^{2lm}(v-2^{2lm}))>v_2(uv-v^2)$, which implies that $N_u$ is even. This leads to a contradiction to $N_u=2^m+1$.

\textbf{Case 2: $ l=2 $}

In this case, $d=\frac{2^{4m}-1}{2^m+1}+2^s$. Let $D_0=\{\al^{j\frac{2^{4m}-1}{2^m+1}} \mid 0 \le j \le 2^m\}$. Suppose $a=\al^{2^m+1}$ and $b=\al^{\frac{2^{4m}-1}{2^m+1}}$. Then $C_0=\langle a \rangle$ and $D_0=\langle b \rangle$. Since $(\frac{2^{4m}-1}{2^m+1},2^m+1)=(2,2^m+1)=1$, any $x \in \GFtm^*$ can be uniquely expressed as $a^ib^k$, for some $0 \le i \le \frac{2^{4m}-1}{2^m+1}-1$ and $0 \le k \le 2^m$. Since
$$
zx+x^{d2^{-s}}=za^ib^k+(a^ib^k)^{d2^{-s}}=a^i(zb^k+b^{kd2^{-s}}),
$$
$zx+x^{d2^{-s}}$ belongs to $C_{\infty}$ (resp. $C_0$, $C_1$) if and only if $zb^k+b^{kd2^{-s}}$ belongs to $C_{\infty}$ (resp. $C_0$, $C_1$). In addition, given two distinct $x_1, x_2 \in A$ with $x_1=a^{i_1}b^{k_1}$ and $x_2=a^{i_2}b^{k_2}$, we have $k_1 \ne k_2$. Otherwise, $x_1x_2^{-1} \in C_0$, which is impossible by the definition of $A$. Thus, for $i=0,1,\infty$,
\begin{align*}
n_i(z) &= |\{x \in A \mid zx+x^{d2^{-s}} \in C_i \}|\\
       &= |\{0 \le k \le 2^m \mid zb^k+b^{kd2^{-s}} \in C_i\}|\\
       &= |\{x \in D_0 \mid zx+x^{d2^{-s}} \in C_i\}|.
\end{align*}
Since $(d2^{-s}-1,2^{4m}-1)=\frac{2^{4m}-1}{2^m+1}$, it is easy to see that $\nin=1$ if and only if $z \in D_0$. Moreover, we have
\begin{align*}
\nz &= |\{x \in D_0 \mid (zx+x^{d2^{-s}})^{\frac{2^{4m}-1}{2^m+1}}=1\}|.
\end{align*}
From now on, we always regard $x$ and $z$ as elements of $D_0$. Remind that $x \in D_0$ if and only if $x^{2^m+1}=1$, the equation
\begin{equation}\label{eqn1}
(zx+x^{d2^{-s}})^{\frac{2^{4m}-1}{2^m+1}}=1
\end{equation}
is equivalent to
\begin{equation}
\left\{\begin{array}{c} zx+x^{d2^{-s}} \ne 0, \\
1+\frac{1}{zx^{d2^{-s}+1}}=0.
\end{array}\right.
\end{equation}
Set $u=(d2^{-s}+1,2^m+1)$, $1+\frac{1}{zx^{d2^{-s}+1}}=0$ has exactly $u$ solutions in $D_0$.

If $u<2^m+1$, since $u$ is a divisor of $2^m+1$, $u \le \frac{2^m+1}{3}$. It is easy to verify that $zx+x^{d2^{-s}}=0$ and $1+\frac{1}{zx^{d2^{-s}+1}}=0$ share one common solution if and only if $z=1$. Hence, we have
$$
n_{\infty}(1)=1, \quad n_0(1)=u-1, \quad n_1(1)=2^m-u+1,
$$
which leads to
\begin{align*}
\ S_d(1) &= \frac{1}{2^m+1}(2^{4m}-2^{3m}(u-1)+2^{2m}(2^m-u+1))\\
    &= 2^{3m}+\frac{2^{3m}}{2^m+1}-2^{2m}u\\
    &\ge 2^{3m}+\frac{2^{3m}}{2^m+1}-2^{2m}\cdot\frac{2^m+1}{3}\\
    &\ge 2^{2m+1}.
\end{align*}
Similarly, if $z \ne 1$, we have
$$
\nin=1, \quad \nz=u, \quad \no=2^m-u,
$$
which implies
\begin{align*}
\Sd &= \frac{1}{2^m+1}(2^{4m}-2^{3m}u+2^{2m}(2^m-u))\\
    &= 2^{3m}-2^{2m}u\\
    &>0.
\end{align*}
Hence, when $u<2^m+1$, $\Sd$ takes at least two distinct positive values and one of which is greater than or equal to $2^{2m+1}$.

If $u=2^m+1$, a similar treatment yields
$$
n_{\infty}(1)=1, \quad n_0(1)=2^m, \quad n_1(1)=0,
$$
which implies $S_d(1)=0$. Meanwhile, for $z \ne 1$, we have
$$
\nin=1, \quad \nz=0, \quad \no=2^m,
$$
which implies $\Sd=2^{3m} \ge 2^{2m+1}$. It is easy to see that $\Sd=2^{3m}$ if and only if $z \in D_0\setminus\{1\}$. Assume $\Sd$ takes three values. Then $2^{3m}$ must be the only positive value that $\Sd$ attains. Suppose $\Sd \in \{2^{3m},v ,0\}$ with $v<0$. Consequently,
\begin{align*}
\sum_{z \in \F_{2^{4m}}} \Sd &= 2^{3m}\cdot2^m+vN_v < 2^{4m},
\end{align*}
which contradicts the first equation of Lemma~\ref{3}.
\end{proof}

\begin{remark}
When $l=2$ and $s=1$, $d=(2^{2m}+1)(2^m-1)+2$ is of Niho type. This decimation has been studied in \cite{YN70} where $\Cd$ takes exactly four values.
\end{remark}

\section{Conclusion}

This paper demonstrates some new results on the cross correlation between an $m$-sequence and its decimated sequence. We make two contributions to this problem. The first one is the determination of the cross correlation distribution for the ternary $m$-sequence with period $3^{3r}-1$ and decimation $d=3^r+2$ or $d=3^{2r}+2$, where $(r,3)=1$. In the case with $(r,3)=3$, it is conjectured that the cross correlation distribution is the same as $(r,3)=1$. The second one is an initial step towards the cross correlation of binary $m$-sequences with period $2^{2lm}-1$ and decimation $d=\frac{2^{2lm}-1}{2^{m}+1}+2^{s}$, where $l \ge 2$ is even and $0 \le s \le 2m-1$. We prove the cross correlation takes at least four values. Additionally, we verify that two famous conjectures due to Sarwate et al. and Helleseth are true in this case. For the cross correlation distribution, numerical experiments show that the cross correlation may take eight or more values. Hence, determining the cross correlation distribution seems to be a very challenging problem.


\begin{thebibliography}{10}
\bibitem{BM10}
N.~Boston and G.~McGuire, ``The weight distributions of cyclic codes with two
  zeros and zeta functions,'' \emph{J. Symbolic Comput.}, vol.~45, no.~7, pp.
  723--733, 2010.

\bibitem{CCD00}
A.~Canteaut, P.~Charpin, and H.~Dobbertin, ``Binary {$m$}-sequences with
  three-valued crosscorrelation: a proof of {W}elch's conjecture,'' \emph{IEEE
  Trans. Inform. Theory}, vol.~46, no.~1, pp. 4--8, 2000.

\bibitem{C04}
P.~Charpin, ``Cyclic codes with few weights and {N}iho exponents,'' \emph{J.
  Combin. Theory Ser. A}, vol. 108, no.~2, pp. 247--259, 2004.

\bibitem{CKN12}
S.~T. Choi, J.~Y. Kim, and J.~S. No, ``On the cross-correlation of a $p$-ary
  $m$-sequence and its decimated sequences by
  $d=\frac{p^{n}+1}{p^{k}+1}+\frac{p^{n}-1}{2}$,'' arXiv:1205.5959.

\bibitem{CN12}
S.~T. Choi and J.~S. No, ``On the cross-correlation distributions between p-ary
  m-sequences and their decimated sequences,'' \emph{IEICE Trans.
  Fundamentals.}, vol. E95-A, no.~11, pp. 1808--1818, 2012.

\bibitem{CD96}
T.~W. Cusick and H.~Dobbertin, ``Some new three-valued crosscorrelation
  functions for binary {$m$}-sequences,'' \emph{IEEE Trans. Inform. Theory},
  vol.~42, no.~4, pp. 1238--1240, 1996.

\bibitem{D98}
H.~Dobbertin, ``One-to-one highly nonlinear power functions on {${\rm
  GF}(2^n)$},'' \emph{Appl. Algebra Engrg. Comm. Comput.}, vol.~9, no.~2, pp.
  139--152, 1998.

\bibitem{DFHR06}
H.~Dobbertin, P.~Felke, T.~Helleseth, and P.~Rosendahl, ``Niho type
  cross-correlation functions via {D}ickson polynomials and {K}loosterman
  sums,'' \emph{IEEE Trans. Inform. Theory}, vol.~52, no.~2, pp. 613--627,
  2006.

\bibitem{DHKM01}
H.~Dobbertin, T.~Helleseth, P.~V. Kumar, and H.~Martinsen, ``Ternary
  {$m$}-sequences with three-valued cross-correlation function: new decimations
  of {W}elch and {N}iho type,'' \emph{IEEE Trans. Inform. Theory}, vol.~47,
  no.~4, pp. 1473--1481, 2001.

\bibitem{FLX13}
T.~Feng, K.~Leung, and Q.~Xiang, ``Binary cyclic codes with two primitive
  nonzeros,'' \emph{Sci. China Math.}, vol.~56, no.~7, pp. 1403--1412, 2013.

\bibitem{G68}
R.~Gold, ``Maximal recursive sequences with 3-valued recursive
  cross-correlation functions (corresp.),'' \emph{IEEE Trans. Inform. Theory},
  vol.~14, no.~1, pp. 154--156, 1968.

\bibitem{GG05}
S.~W. Golomb and G.~Gong, \emph{Signal design for good correlation: For
  wireless communication, cryptography, and radar}.\hskip 1em plus 0.5em minus
  0.4em\relax Cambridge: Cambridge University Press, 2005.

\bibitem{TH76}
T.~Helleseth, ``Some results about the cross-correlation function between two
  maximal linear sequences,'' \emph{Discrete Math.}, vol.~16, no.~3, pp.
  209--232, 1976.

\bibitem{TH1978}
------, ``A note on the cross-correlation function between two binary maximal
  length linear sequences,'' \emph{Discrete Math.}, vol.~23, no.~3, pp.
  301--307, 1978.

\bibitem{TH03}
------, ``Pairs of {$m$}-sequences with a six-valued crosscorrelation,'' in
  \emph{Mathematical properties of sequences and other combinatorial structures
  ({L}os {A}ngeles, {CA}, 2002)}.\hskip 1em plus 0.5em minus 0.4em\relax
  Boston, MA: Kluwer Acad. Publ., 2003, pp. 1--6.

\bibitem{HHKZLJ09}
T.~Helleseth, L.~Hu, A.~Kholosha, X.~Zeng, N.~Li, and W.~Jiang,
  ``Period-different {$m$}-sequences with at most four-valued cross
  correlation,'' \emph{IEEE Trans. Inform. Theory}, vol.~55, no.~7, pp.
  3305--3311, 2009.

\bibitem{HK98}
T.~Helleseth and P.~V. Kumar, ``Sequences with low correlation,'' in
  \emph{Handbook of coding theory, {V}ol. {I}, {II}}.\hskip 1em plus 0.5em
  minus 0.4em\relax Amsterdam: North-Holland, 1998, pp. 1765--1853.

\bibitem{HR05}
T.~Helleseth and P.~Rosendahl, ``New pairs of {$m$}-sequences with 4-level
  cross-correlation,'' \emph{Finite Fields Appl.}, vol.~11, no.~4, pp.
  674--683, 2005.

\bibitem{HX01}
H.~D.~L. Hollmann and Q.~Xiang, ``A proof of the {W}elch and {N}iho conjectures
  on cross-correlations of binary {$m$}-sequences,'' \emph{Finite Fields
  Appl.}, vol.~7, no.~2, pp. 253--286, 2001.

\bibitem{JH09}
A.~Johansen and T.~Helleseth, ``A family of {$m$}-sequences with five-valued
  cross correlation,'' \emph{IEEE Trans. Inform. Theory}, vol.~55, no.~2, pp.
  880--887, 2009.

\bibitem{JHK09}
A.~Johansen, T.~Helleseth, and A.~Kholosha, ``Further results on
  {$m$}-sequences with five-valued cross correlation,'' \emph{IEEE Trans.
  Inform. Theory}, vol.~55, no.~12, pp. 5792--5802, 2009.

\bibitem{K71}
T.~Kasami, ``The weight enumerators for several classes of subcodes of the 2nd
  order binary reed-muller codes,'' \emph{Inf. Control}, vol.~18, no.~4, pp.
  369 -- 394, 1971.

\bibitem{K12}
D.~J. Katz, ``Weil sums of binomials, three-level cross-correlation, and a
  conjecture of {H}elleseth,'' \emph{J. Combin. Theory Ser. A}, vol. 119,
  no.~8, pp. 1644--1659, 2012.

\bibitem{LN83}
R.~Lidl and H.~Niederreiter, \emph{Finite fields}, ser. Encyclopedia of
  Mathematics and its Applications.\hskip 1em plus 0.5em minus 0.4em\relax
  Reading, MA: Addison-Wesley Publishing Company Advanced Book Program, 1983,
  vol.~20.

\bibitem{LF08}
J.~Luo and K.~Feng, ``Cyclic codes and sequences from generalized
  {C}oulter-{M}atthews function,'' \emph{IEEE Trans. Inform. Theory}, vol.~54,
  no.~12, pp. 5345--5353, 2008.

\bibitem{NH06}
G.~J. Ness and T.~Helleseth, ``Cross correlation of {$m$}-sequences of
  different lengths,'' \emph{IEEE Trans. Inform. Theory}, vol.~52, no.~4, pp.
  1637--1648, 2006.

\bibitem{NH062}
------, ``A new three-valued cross correlation between {$m$}-sequences of
  different lengths,'' \emph{IEEE Trans. Inform. Theory}, vol.~52, no.~10, pp.
  4695--4701, 2006.

\bibitem{NH07}
------, ``A new family of four-valued cross correlation between {$m$}-sequences
  of different lengths,'' \emph{IEEE Trans. Inform. Theory}, vol.~53, no.~11,
  pp. 4308--4313, 2007.

\bibitem{YN70}
Y.~Niho, ``Multivalued cross-correlation functions between two maximal linear
  recursive sequence,'' Ph.D. dissertation, Univ. Southern Calif., Los Angeles,
  1970.

\bibitem{SP80}
D.~Sarwate and M.~Pursley, ``Crosscorrelation properties of pseudorandom and
  related sequences,'' \emph{Proceedings of the IEEE}, vol.~68, no.~5, pp.
  593--619, 1980.

\bibitem{SKNS08}
E.~Y. Seo, Y.~S. Kim, J.~S. No, and D.~J. Shin, ``Cross-correlation
  distribution of {$p$}-ary {$m$}-sequence of period {$p^{4k}-1$} and its
  decimated sequences by {$(\frac{p^{2k}+1}2)^2$},'' \emph{IEEE Trans. Inform.
  Theory}, vol.~54, no.~7, pp. 3140--3149, 2008.

\bibitem{XZH10}
Y.~Xia, X.~Zeng, and L.~Hu, ``Further crosscorrelation properties of sequences
  with the decimation factor {$d=\frac{p^n+1}{p+1}-\frac{p^n-1}{2}$},''
  \emph{Appl. Algebra Engrg. Comm. Comput.}, vol.~21, no.~5, pp. 329--342,
  2010.
\end{thebibliography}
\end{document}